\documentclass{amsart}
\usepackage[T1]{fontenc}
\usepackage[latin9]{inputenc}
\usepackage{float}
\usepackage{amsthm}
\usepackage{amsmath}
\usepackage{amssymb}
\usepackage{graphicx}
\usepackage{esint,bbm,enumerate,color,fullpage}
\usepackage[unicode=true,
 bookmarks=true,bookmarksnumbered=true,bookmarksopen=true,bookmarksopenlevel=1,
 breaklinks=false,pdfborder={0 0 0},backref=false,colorlinks=false]
 {hyperref}
\hypersetup{pdftitle={Your Title},
 pdfauthor={Your Name},
 pdfpagelayout=OneColumn,pdfnewwindow=true,pdfstartview=XYZ,plainpages=false}

\newcommand{\one}{\mathbbm{1}}
\newcommand{\R}{\mathbb{R}}
\newcommand{\BPDN}{\rm{BPDN}}
\newcommand{\x}{\mathbf{x}}
\newcommand{\X}{\mathbf{X}}

\newcommand{\TPR}{\text{TPR}}
\makeatletter

\providecommand{\tabularnewline}{\\}
\floatstyle{ruled}
\newfloat{algorithm}{tbp}{loa}
\providecommand{\algorithmname}{Algorithm}
\floatname{algorithm}{\protect\algorithmname}

 \let\oldforeign@language\foreign@language
 \DeclareRobustCommand{\foreign@language}[1]{%
   \lowercase{\oldforeign@language{#1}}}
\theoremstyle{plain}
\newtheorem{thm}{\protect\theoremname}
\newtheorem{prop}[thm]{\protect\propname}
\theoremstyle{definition}
\newtheorem{defn}[thm]{\protect\definitionname}
\newtheorem{assumption}{\protect\assumptionname}
\newtheorem{remark}{Remark}
\theoremstyle{plain}
\theoremstyle{plain}
\newtheorem{cor}[thm]{\protect\corollaryname}

\providecommand{\corollaryname}{Corollary}
\providecommand{\propname}{Proposition}
\providecommand{\definitionname}{Definition}
\providecommand{\theoremname}{Theorem}
\providecommand{\assumptionname}{Assumption}

\makeatother

\providecommand{\corollaryname}{Corollary}
\providecommand{\definitionname}{Definition}
\providecommand{\theoremname}{Theorem}

\begin{document}

\title{Sparse regression with highly correlated predictors}

\author{Behrooz~Ghorbani, {\"O}zg{\"u}r~Y{\i}lmaz}
\thanks{Behrooz~Ghorbani is with the Electrical Engineering
  Department, Stanford University, Palo Alto, CA, USA, e-mail: \protect\protect\href{http://b.ghorbani.bg@gmail.com}{b.ghorbani.bg@gmail.com}.%
}%
\thanks{{\"O}zg{\"u}r~Y{\i}lmaz is with the Department of Mathematics,
  University of British Columbia, Vancouver, BC, Canada, e-mail: \protect\protect\href{http://xxx@xxx.xxx}{oyilmaz@math.ubc.ca}.%
}

\maketitle
\begin{abstract}
  We consider a linear regression $y=X\beta+u$ where
  $X\in\mathbb{\mathbb{{R}}}^{n\times p},\; p\gg n,$ and $\beta$ is
  $s-$sparse. Motivated by examples in financial and economic data, we
  consider the situation where $X$ has highly correlated and clustered
  columns. To perform sparse recovery in this setting, we introduce
  the \emph{clustering removal algorithm} (CRA), that seeks to
  decrease the correlation in $X$ by removing the cluster structure
  without changing the parameter vector $\beta$. We show that as long
  as certain assumptions hold about $X$, the decorrelated matrix will
  satisfy the restricted isometry property (RIP) with high
  probability. We also provide examples of the empirical performance
  of CRA and compare it with other sparse recovery
  techniques. \end{abstract}

\section{Introduction}

We consider the \emph{sparse estimation} (or \emph{sparse recovery}) problem given by: 
\begin{equation}
\mbox{ Estimate \ensuremath{\beta}\ from } y=X\beta+u\label{prob}
\end{equation}
where ${X}$ is a known $n\times p$ design matrix, $u$ is additive noise
such that $\|u\|_2\le \eta$ for some known $\eta$,
and $\beta$ has at most $s$ non-zero entries with $s\ll n \ll
p$. Lasso \cite{tibshirani1996regression} and basis pursuit denoise
(BPDN) \cite{Donoho2004} 
are popular sparse recovery approaches proposed for \eqref{prob}, which are based on
solving certain convex optimization problems. On the other hand,
orthogonal matching pursuit (OMP) \cite{4385788}, CoSaMP
\cite{needell2009cosamp}, Iterative Hard Thresholding \cite{BD08}, and
Hard Thresholding Pursuit \cite{doi:10.1137/100806278} are examples of
greedy algorithms that can be used to solve \eqref{prob}. Although
these algorithms have different capabilities in estimating $\beta$,
the success of each is highly dependent on the structure of the design
matrix $X$ and the sparsity level $s$ in relation to $p$ and $n$. It
is well known that when the columns of $X$ have high empirical
correlation, i.e., for some $i\neq j$,
$\rho_{ij}:=\dfrac{X_{i}X_{j}^{*}}{||X_{i}||_{2}||X_{j}||_{2}}$ has
magnitude close to one, the above mentioned sparse recovery methods
do not perform well \cite{2012arXiv1209.5908B}.

In many applications, such as in model selection, $X$ is
already given in terms of observations of different variables of the
data. Often, the variables constituting the columns of $X$
are highly correlated. Therefore, at least for a significant number of
columns $i,j$, $|\rho_{ij}|$ tends to be close to 1, and accordingly,
standard sparse recovery methods, such as the ones we mention above,
fail to yield good estimates for $\beta$. 

In this note, we will focus on this issue, i.e., sparse estimation
when the design matrix has correlated columns. Motivated by various
applications where $X$ is empirically generated, we will further
assume that the columns of $X$ can be organized in a number of
clusters such that the columns that are in the same cluster are highly
correlated but the columns that are in different clusters may or may
not be correlated.  For example, let $X_{t,i}$ be the price of stock
$i$ at time $t$. Due to shared underlying economic factors, we expect
the price vector of, for example, technology stocks to be tightly
clustered together. On the other hand, stocks of telecommunication
companies may form a different cluster. Due to global economic
factors, such as the monetary policy or the overall economic growth,
the stock prices of telecommunication companies and technology
companies may also be correlated, but we expect lower levels of
correlation between them compared to the correlation inside the
clusters. Another example where such cluster structure can be observed
is in face recognition -- see \cite{wright2010dense} for details.


Our approach can be summarized as follows: To overcome the challenge
posed by high correlation in the design matrix, we propose to modify
$X$ in order to get a matrix that is suitable for sparse recovery
without changing the original parameter vector $\beta$.  We seek to do
this by first identifying the clusters (say, $q$ of them) and then constructing a
representative vector for each cluster. Let $R\in\mathbb{R}^{n\times
  q}$ be the matrix constructed by putting together the representative
vectors of the $q$ clusters. We project $X$ to the orthogonal
complement of the range of $R$ and normalize the columns of the
resulting matrix, which we denote by $\tilde{X}$. Our proposed
algorithm, we will prove, is effective when this matrix $\tilde{X}$ is
suitable as a compressive sensing measurement matrix.

After introducing our notation and the necessary background in
Section~\ref{notations}, we state the proposed clustering removal
algorithm (CRA) and our main assumptions in Section~\ref{alg}. In
Section~\ref{sec:Theoretical-Discussions}, we prove that if $X$ is a
realization of a random matrix $\mathbf{X}$ whose columns are
uniformly distributed on $q$ disjoint spherical caps, $\tilde{{X}}$
provides a suitable matrix for sparse recovery with high
probability. In Section \ref{sec:Numerical-Results} we demonstrate the
performance of our algorithm on highly correlated financial data, in
comparison with BPDN and with SWAP, an algorithm proposed by
\cite{vats2013swapping} for sparse recovery in highly correlated
settings. 

\section{Notations and Background}\label{notations}

In what follows, we refer to random matrices and random vectors with
bold letters.  Let $n,r\in \mathbb{N}$ and let $A\in\mathbb{R}^{n\times
  r}$. We denote the $i$th column of $A$ by $A_{i}$. For $T \subseteq
[r]:=\{1,\dots,r\}$, $A_T$ denotes the submatrix of $A$ consisting of
its columns indexed by $T$. We define $\Pi_{A}$ to be the
orthogonal projection operator into $\mathcal{R}(A)$, the range of
$A$, and $\Pi_{A^\perp}$ denotes the orthogonal projection operator
into the orthogonal complement of $\mathcal{R}(A)$. The best $k$-term
approximation of $b\in\mathbb{R}^{n}$ is $b^{\{k\}}\in\mathbb{R}^{n}$
such that $b_{i}^{\{k\}}=b_{i}$ if $b_{i}$ is among the $k$
largest-in-magnitude entries of $b$. Otherwise, $b_{i}^{\{k\}}=0$. The
corresponding best $k$-term approximation error in $\ell_p$ is
$$\sigma_k(b)_{\ell_p}:=\|b-b^{\{k\}}\|_p.$$
For $\Omega\subset
S^{n-1}:=\{x\in \R^n: \|x\|_2=1\}$, we define the \emph{wedge}
$W(\Omega)$ by 
$$W(\Omega):=\{rx: x\in\Omega, r\in[0,1]\}.$$  
We denote the
$n$-dimensional Lebesgue measure with $\lambda^{n}$, and the uniform
spherical measure on $S^{n-1}$ with $\sigma^{n-1}$, which can be
defined via
$$\sigma^{n-1}(\Omega):=\lambda^{n}(W(\Omega))/\lambda^{n}(W(S^{n-1}))$$
provided that $W(\Omega)$ is a measurable subset of $\R^n$ (otherwise
$\Omega$ is not measurable).  It should be noted that although
different definitions of uniform spherical measure exist
in the literature, all these measures must be equal, for example, in
the setting that we have described above \cite{christensen1970some}.

\subsection{Compressive sampling}
It is well known in the compressive sampling literature that BPDN
provides a stable and robust solution for the sparse estimation problem
\eqref{prob} in a computationally tractable way , e.g.,
\cite{Donoho2004,CandesRombergTao2006}. Specifically, BPDN provides
the estimate
\begin{equation}\label{eq:BPDN}
\beta_{\BPDN(y,\; X,\;\eta)}^{\#}:=
\arg\min_{z}||z||_{1}\; s.t\quad||Xz-y||_{2}\leq\eta
\end{equation}
where $\eta$ is an upper bound on $||u||_{2}$. One can guarantee that
$\beta_{\BPDN(y,\; X,\;\eta)}^{\#}$ is a good estimate of $\beta$ if
$\beta$ is sufficiently sparse or \emph{compressible} (i.e., it can be well
approximated by a sparse vector) and the ``restricted isometry
constants'' of $X$ are sufficiently small
\cite{CandesRombergTao2006}, where the restricted isometry constants are defined as follows.
\begin{defn}
Let $\delta_{s}(X)$ be the smallest positive constant such that
$\forall T\subset [p],$
$|T|=s$, $\forall z\in\mathbb{R}^{s}$ 
\[
(1-\delta_{s}(X))||z||_{2}^{2}\leq||X_{T}z||_{2}^{2}\leq(1+\delta_{s}(X))||z||_{2}^{2}.
\]
In this case, we say that $X$ satisfies the
restricted isometry property (RIP) of order $s$ with (restricted
isometry) constant $\delta_s$.
\end{defn}
The following theorem by \cite{caiRIP} is a sharper version of the
original ``stable and robust recovery theorem'' of
\cite{CandesRombergTao2006}.
\begin{thm} $\label{th:estimation}$
Let $\delta_{ts}(X)<\sqrt{\frac{t-1}{t}}$ for some $t\geq\dfrac{4}{3}$.
Then, for any $s$-sparse $\beta\in\mathbb{R}^{p}$, 
\[
||\beta_{\BPDN(y,\; X,\;\eta)}^{\#}-\beta||_{2}\leq \dfrac{C\sigma_{s}(\beta)_{\ell_1}}{\sqrt{s}}+D\eta
\]
where $\beta_{\BPDN(y,\; X,\;\eta)}$ is as in \eqref{eq:BPDN}, and $C,D>0$ are constants depending only on $\delta_{ts}$.
\end{thm}

Even though the above theorem provides guarantees for BPDN to recover
(or estimate) a sparse or compressible vector $\beta$ from its
(possibly noisy) compressive samples $y$, these guarantees depend on
$X$ having small restricted isometry constants. Constructing matrices
with (nearly) optimally small restricted isometry constants is an
extremely challenging task and still an open problem. Furthermore,
computing these constants entails a combinatorial computational
complexity and is not tractable as the size of the matrix
increases. As a remedy, the literature has focused on random
matrices. Indeed, one can prove that certain classes of sub-Gaussian
random matrices (see next section) have nearly optimally small
restricted isometry constants with overwhelming probability. Accordingly,
realizations of such random matrices are used in the context of
compressed sensing. We will adopt such an approach.

\subsection{Sub-Gaussian random matrices}
Our theoretical analysis in Section~\ref{sec:Theoretical-Discussions}
relies on certain fundamental properties of sub-Gaussian random
matrices. Here, we state some basic definitions that we will need. See
\cite{vershynin2010introduction} for a thorough exposition
on the non-asymptotic theory of sub-Gaussian random matrices. In what
follows, we stick to the notation of \cite{vershynin2010introduction}.

\begin{defn} A random variable $\x$ is said to be a sub-Gaussian random
  variable if it satisfies 
$$\big(\mathbb{E}|\x|^q\big)^{1/q}\le K \sqrt{q}$$
for all $q \ge 1$. Its sub-Gaussian norm $\|\x\|_{\psi_2}$ is defined
as 
$$\|\x\|_{\psi_2}=\sup_{q\ge 1}q^{-1/2}\big(\mathbb{E}|\x|^q\big)^{1/q}.$$
\end{defn}

Next we define the following two important classes of random vectors. 

\begin{defn} Let $\x$ be a random vector in $\R^n$.
\begin{enumerate}[(i)]
\item $\x$ is said to be \emph{isotropic} if $\mathbb{E}\langle
  \x,u\rangle^2 = \|u\|^2$ for all $u\in \R^n$.
\item $\x$ is sub-Gaussian if the marginals $\langle \x,u\rangle$ are
sub-Gaussian random variables for all $u\in \R^n$. The sub-Gaussian
norm of the random vector $\x$ is defined by
$$\|\x\|_{\psi_2}:=\sup_{u\in S^{n-1}} \|\langle \x,u\rangle\|_{\psi_2}.$$

\end{enumerate}
\end{defn}


\section{Assumptions and the Algorithm}\label{alg}
We now go back to the sparse estimation problem \eqref{prob}, i.e., we
want to estimate $\beta$ from $y=X\beta+u$ when given the matrix $X$
and the fact that $\|u\|_2\eta$.  For the purpose of this paper, we
assume that $||X_{i}||_{2}=1$ $\forall i\in[p]$. This simply can be
done by normalizing the columns and rescaling $\beta$.  Below, $C_z$
denotes a spherical cap, i.e., a portion of $S^{n-1}$ cut off by some
hyperplane $P$, with ``centroid'' $z$, which is the point on the cap
that has maximum distance from $P$.

We make the following additional assumptions.

\begin{assumption}\label{a0} The data matrix $X$ is a realization of a
  random matrix $\X$. 
\end{assumption}

\begin{assumption} \label{a1} There exist
  $z_{1},\dots,z_{q}\in S^{n-1}$, $q\ll n$,
such that $\forall i\in[p],\;\exists j\in[q]$ s.t $\mathbf{X_{i}}$
is uniformly distributed on $C_{z_{j}}$. That is, $\mathbf{X_{i}}$ is
distributed according to the measure $m_{C_{z_{j}}}$, where  for all
measurable subsets $A$ of $C_{z_{j}}$,
$$m_{C_{z_{j}}}(A):=\lambda^{n}(W(A))/\lambda^{n}(W(C_{z_{j}})).
$$
In this case, we write $\mathbf{X}_{i}\sim C_{z_{j}}$. 
\end{assumption}

\begin{assumption} \label{a2} For $i\neq j$, $\mathbf{X_{i}}$ and
  $\mathbf{X_{j}}$ are independent.
\end{assumption}
Clearly, by Assumption~\ref{a1}, $X_{i}$ can be clustered into $q$ groups
$G_1,\dots,G_q$ where 
\[
G_{j}=\{i\in[p]\; s.t.\;\mathbf{X}_{i}\sim C_{z_{j}}\}
\]
\begin{algorithm}[*tb] 
\protect\protect\caption{Clustering Removal Algorithm (CRA) \label{alg1}}

\begin{description}
\item[Step 1:] \ Estimate $G_{1},G_{2},...,G_{q}$.
\smallskip 
\item[Step 2:] \ Estimate $R:=[z_{1}|z_{2}|...|z_{q}]$ by setting
$\hat{{z_{i}}}={\displaystyle \frac{1}{|G_{i}|}\sum_{j\in
    G_{i}}X_{j}}$. 
\smallskip
\item[Step 3:] \ Use a sparse recovery method for obtaining an
  estimate $\hat{\gamma}$ for $\gamma$ from 
\begin{equation}
\tilde{y}=\tilde{X}\gamma+\tilde{u}\label{secondary}
\end{equation}
where $\tilde{y}=\Pi_{R^{\perp}}y$,
$\tilde{X}=\Pi_{R^{\perp}}XN_{\Pi_{R^{\perp}}X}^{-1}$,
$\tilde{u}=\Pi_{R^{\perp}}u$.
\smallskip
\item[Step 4:] \ The estimated $\beta$ will be
  $\hat{{\beta}}:=N_{\Pi_{R^{\perp}}X}^{-1}\hat{\gamma}$
\end{description}
\end{algorithm}

Under these assumptions, we propose the clustering removal algorithm
(CRA), described in Algorithm~\ref{alg1}, for estimating $\beta$ in
\eqref{prob}. In the first step of CRA, we estimate $G_j$, $j\in [q]$
by clustering the columns of $X$ using an appropriate clustering
algorithm. For the numerical simulations in this paper, we use
$k$-means clustering. Note that here we make the additional assumptions
that we know $q$ and it is possible to estimate $G_j$, which, for
example, requires that the spherical caps $C_{z_j}$ are well
separated. In step 2, we estimate the centroid $z_j$ of each cluster
$G_j$ by averaging the columns of $X$ that belong to $G_j$. Note that
this is an unbiased estimate for $z_j$ assuming the clusters have been
accurately identified. In step 3, we project each column of $X$ onto
the orthogonal complement of the $q$-dimensional subspace spanned by
the cluster centroids (without loss of generality, we assume that the
set of cluster centroids $\{z_1,\dots,z_q\}$ is linearly independent).
%
We normalize the projected columns by multiplying the matrix with a diagonal matrix $N_{\Pi_{R^{\perp}}X}^{-1}$ where $(N_{\Pi_{R^{\perp}}X})_{i,i}=||\Pi_{R^{\perp}}X_i||_2$. After normalization, we obtain our modified matrix
$\widetilde{X}:=\Pi_{R^{\perp}}XN_{\Pi_{R^{\perp}}X}^{-1}$ which, as
we will show in the next section, is an appropriate measurement matrix
for sparse recovery. The main idea here is that the columns of
$\widetilde{X}$ lie on the $(n-q)$-dimensional subspace
$\mathcal{R}(R)^\perp$ and, in fact, we show below that
$\widetilde{X}$ is a realization of $\widetilde{\X}:=\Pi_{R^{\perp}}\X
N_{\Pi_{R^{\perp}}\X}^{-1}$ where
$\widetilde{\mathbf{X}_i}$ are independent and uniformly distributed
on $S^{n-1}\cap \mathcal{R}(R)^\perp$. Accordingly, $\widetilde{X}$
turns out to be a good compressive sampling matrix, as explained in
the next section, and we use a sparse recovery algorithm of our choice
to obtain $\hat{\gamma}$, an estimate of $\gamma$, from
\eqref{secondary}. In the final step, we ``unnormalize'' by
multiplying $\hat{\gamma}$ with the diagonal matrix
$N_{\Pi_{R^{\perp}}X}^{-1}$.


\section{Theoretical Recovery Guarantees for CRA}  \label{sec:Theoretical-Discussions}

In this section, we assume that the clusters $G_j$ and their centroids
$z_j$, $j=1,\dots,q$ are known (or accurately estimated). Under this
assumption, we show that $\widetilde{\mathbf{X}}$, with high
probability, is an appropriate compressive sampling matrix. The
following two theorems from \cite{vershynin2010introduction} will be
instrumental in our theoretical analysis of CRA.
\begin{thm}\label{th11}
  \cite[Theorem 5.65]{vershynin2010introduction} Let $\mathbf{Z}$ be
  an $n\times p$ random matrix with $n< p$ and
  $||\mathbf{Z_{i}}||_{2}=1$ for all $i\in[p]$. Let $k\in[p]$ and
  $\delta\in(0,1).$ Suppose that
\begin{enumerate}[(i)]
\item the columns of $\mathbf{Z}$ are independent, and 
\item the columns of $\sqrt{n}\mathbf{Z}$ are sub-Gaussian and isotropic.
\end{enumerate}
Then, there exists positive constants $c$ and $C$, depending only on
the maximum sub-Gaussian norm $K:=\max_{i}\|\mathbf{X}_i\|_{\psi_2}$
of the columns of $\mathbf{Z}$, such that

\begin{equation} n\geq C\delta^{-2}k\log(\frac{ep}{k}) \implies 
\delta_{k}(\mathbf{Z})\leq\delta
\end{equation} 
with probability at least
$1-2\exp(-c\delta^{2}n)$.
\end{thm}
\begin{thm}\label{th22}
\cite{vershynin2010introduction} Let $\mathbf{z}$ be
an $n\times1$ random vector. If $\mathbf{z}$ is uniformly distributed
on $\sqrt{n}S^{n-1},$ then $\mathbf{z}$ is both sub-Gaussian
and isotropic. Furthermore, the sub-Gaussian norm of $\mathbf{z}$ is
bounded by a universal constant.
\end{thm}
Next, we will prove that the columns of $\widetilde{\mathbf{X}_i}$
satisfy the hypotheses of Theorem~\ref{th11}. To that end, for
$z\in\mathbb{R}^{n}$ define
$$\omega(z) :=\begin{cases}
0_{n\times 1} & \text{if $\Pi_{R^{\perp}}z=0$,} \\
\frac{\Pi_{R^{\perp}}z}{\left \| \Pi_{R^{\perp}}z\right\| _{2}} &
\text{otherwise.}
\end{cases}
$$
Note that $\widetilde{\mathbf{X}_i}=\omega(\mathbf{X}_i)$.

\begin{thm}
  If Assumption~\ref{a1} and Assumption~\ref{a2} hold,
  $\omega(\mathbf{X}_{i})$ is uniformly distributed on $S^{n-1}\cap
  \mathcal{R}(R)^\perp$ for each $i\in[n]$. Moreover, for $i\neq j,$
  $\omega(\mathbf{X}_{i})$ and $\omega(\mathbf{X}_{j})$ are
  independent.\label{thmmain}\end{thm}
\begin{proof}
  Since $\mathbf{X}_{i}$ and $\mathbf{X}_{j}$ are independent, and
  $\omega:\R^n\mapsto\R^n$ is measurable, $\omega(\mathbf{X}_{i})$ and
  $\omega(\mathbf{X}_{j})$ are also independent. Next we show that
  $\omega(\mathbf{X}_{i})$ is uniformly distributed on $S^{n-1}\cap
  \mathcal{R}(R)^\perp$. To that end, let $\mathbf{X}_{i}\in
  C_{R_{1}}$ and let $P$ be the hyperplane that generates $C_{R_{1}}.$
  Then, $R_1$ is a normal vector of $P$ and accordingly we have
  $$P=\{x\in\mathbb{R}^{n}:\langle x-q,  R_{1}\rangle=0\}$$
  where $q$ is an arbitrary point on $P$ that we fix. Recall, on the
  other hand, that 
\begin{equation}
C_{R_1}=\{x\in S^{n-1}: \langle x,R_{1}\rangle \geq \langle q,R_{1}\rangle\}\label{character}
\end{equation}
where $q\in P$ is as above. Note that $\langle q,R_{1} \rangle$ is independent
of the choice of $q$.

Now, let $\Omega\subset S^{n-1}\cap \mathcal{R}(R)^\perp$. Noting that
$||\mathbf{X_{i}}||_{2}=1$, we observe that 
\begin{align}
\mathbb{P}(\omega(\mathbf{X_{i}})\in\Omega)&=\mathbb{P}(\Pi_{R^{\perp}}\mathbf{X_{i}}\in W(\Omega)\setminus \{\mathbf{0}\}) \\
&= \mathbb{P}(\Pi_{R^{\perp}}\mathbf{X_{i}}\in W(\Omega))-\mathbb{P}(\Pi_{R^{\perp}}\mathbf{X_{i}}=\mathbf{0})\label{eq:second_prob}
\end{align}
since $\mathbf{0}\in W(\Omega)$. We know that any vector
$z\in\mathbb{R}^{n}$ can be uniquely written as $z_{R}+z_{R^{\perp}}$
where $z_{R}\in \mathcal{R}(R)$ and $z_{R^{\perp}}\in
\mathcal{R}(R)^{\perp}$. Note that
$\Pi_{R^{\perp}}z=\Pi_{R^{\perp}}z_{R}+\Pi_{R^{\perp}}z_{R^{\perp}}=z_{R^{\perp}}$.
Therefore, $\Pi_{R^{\perp}}z\in W(\Omega)$ if and only if
$z_{R^{\perp}}\in W(\Omega)$ which is holds if and only if $z$ is an
element of
$$
\Pi^{-1}(\Omega):=\{ru+Rv:\ r\in[0,1],\; u\in\Omega,\;
v\in\mathbb{R}^{q}\}.
$$  
Hence, keeping in mind that
$\mathbf{X}_i \in C_{R_1}$, we have
$$\mathbb{P}(\Pi_{R^{\perp}}\mathbf{X_{i}}\in W(\Omega)) =
\mathbb{P}(\mathbf{X_{i}}\in\Pi^{-1}(\Omega)\cap C_{R_{1}})$$
and
$$\mathbb{P}(\Pi_{R^{\perp}}\mathbf{X_{i}}=\mathbf{0}) =
\mathbb{P}(\mathbf{X_{i}\in}\mathcal{R}(R)\cap C_{R_{1}}).
$$
Furthermore, since the distribution of $\mathbf{X_{i}}$ is continuous,
and $\mathcal{R}(R)$ has a lower dimension compared to $C_{R_{1}}$,
\mbox{$\mathbb{P}(\mathbf{X_{i}\in}C_{R_{1}}\cap
\mathcal{R}(R))=0$}. Accordingly, it follows from
(\ref{eq:second_prob}) that
\begin{equation}\label{sec_prob}
\mathbb{P}(\omega(\mathbf{X_{i}})\in\Omega)=\mathbb{P}(\mathbf{X_{i}}\in\Pi^{-1}(\Omega)\cap
C_{R_{1}}).
\end{equation}

Next, recall that $\mathbf{X}_i$ is uniformly distributed on
$C_{R_1}$, i.e., $\mathbf{X}_i$ is distributed according to the
measure $m_{C_{R_{1}}}$, defined as in Assumption~\ref{a1}. Therefore,
we can rewrite \eqref{sec_prob} as
\begin{align}
\mathbb{P}(\omega(\mathbf{X_{i}})\in\Omega)&=\int\limits_{\Pi^{-1}(\Omega)\cap
  C_{R_{1}}} d m_{C_{R_{1}}}
\notag \\
&= \int\limits_{W(\Pi^{-1}(\Omega)\cap
  C_{R_{1}})} \frac{d\lambda^n(y)}{\lambda^n(W(C_{R_1}))} \notag \\
&= \int\limits_{\Pi^{-1}(\Omega)\cap
  W(C_{R_{1}})} \frac{d\lambda^n(y)}{\lambda^n(W(C_{R_1}))} \notag \\
&=\int\limits_{\Pi^{-1}(\Omega))}\one_{W(C_{R_1})}(y)\dfrac{1}{\lambda^{n}(W(C_{R_{1}}))}d\lambda^{n}(y) \label{Int_1}
\end{align}
where $\one_{W(C_{R_1})}$ is the indicator function of the set
$W(C_{R_{1}})$, i.e., 
$$\one_{W(C_{R_1})}(y):=\begin{cases}
1 & \mbox{if \ensuremath{y\in W(C_{R_{1}})}}\\
0 & \text{otherwise}
\end{cases}.
$$
Above, the second equality follows from the definition of the measure
$m_{C_{R_{1}}}$ and the third equality holds because
$W(\Pi^{-1}(\Omega)\cap C_{R_{1}})=\Pi^{-1}(\Omega)\cap W(C_{R_{1}})$,
which is easy to verify.

Next, we decompose $\Pi^{-1}(\Omega)=W(\Omega) \times \mathcal{R}(R)$ and
denote by $\lambda^q$ and $\lambda^{n-q}$ the Lebesgue measure on
$\mathcal{R}(R)$ and on $\mathcal{R}(R)^\perp$ respectively. Note that
$W(\Omega)$ and $\mathcal{R}(R)$, and thus every $y\in
\Pi^{-1}(\Omega)$ can be uniquely decomposed as $y=y_{R^\perp} +y_R$
such that $y_{R^\perp}\in W(\Omega)$ and $y_R \in \mathcal{R}(R)$. 
Then,
we obtain
\begin{align}
\mathbb{P}(\omega(\mathbf{X_{i}})\in\Omega) \hskip -1.9cm & \notag \\
&=\int\limits_{W(\Omega)}\left[\int\limits_{\mathcal{R}(R)}\frac{\one_{W(C_{R_1})}(y_{R^\perp}+y_R)}{\lambda^{n}(W(C_{R_{1}}))}\;
  d\lambda^{q}(y_R)\right]\;{ {
    d\lambda^{n-q}}(y_{R^\perp})} \notag \\
&={\scriptstyle
  \int\limits_{\Omega}\int\limits_{[0,1]}\left[\int\limits_{\mathcal{R}(R)}\frac{\one_{W(C_{R_1})}(ru+y_R)}{\lambda^{n}(W(C_{R_{1}}))}d\lambda^{q}(y_R)\right]r^{{\scriptstyle
      n-q-1}}dr\;{\scriptstyle d\sigma^{n-q-1}(u)}} \label{ind}
\end{align}
Above, first equality follows from Fubini-Tonelli's theorem, and the
second equality is obtained by passing to polar coordinates and using
Fubini-Tonelli's theorem one more time. Finally, $\sigma^{n-q-1}$ is the
unique spherical measure on $S^{n-1}\cap \mathcal{R}(R)^\perp$, i.e.,
the unit sphere of the subspace $\mathcal{R}(R)^\perp$. 

Next, we observe that $\one_{W(C_{R_1})}(ru+y_R)=1$ if and only if one
of the following two statements hold: 
\begin{enumerate}[(i)]
\item $ru+y_R=0$, which, since $u$
is orthogonal to $y_R$, is
equivalent to having $r^2+\|y_R\|_2^2=0$, or 
\item $0<r^2+\|y_R\|_2^2\le 1$ and 
$$\frac{ru+v}{||ru+v||_{2}}\in C_{R_{1}}$$
which, by \eqref{character}, holds if and only if 
$$\langle y_R,R_{1}\rangle \geq \big
(r^2+\|y_R\|_2^2\big)^{1/2}\langle q,R_{1}\rangle$$
where $q$ is as in \eqref{character}.
\end{enumerate}
Accordingly, we conclude that the innermost integral in \eqref{ind}
does not depend on $u$, which implies that 
\begin{equation}\label{eqinner}
\int\limits_{[0,1]}\left[\int\limits_{\mathcal{R}(R)}\frac{\one_{W(C_{R_1})}(ru+y_R)}{\lambda^{n}(W(C_{R_{1}}))}d\lambda^{q}(y_R)\right]r^{{\scriptstyle
      n-q-1}}dr = c,
\end{equation}
where $c$ is a constant that only depends on $n$. Substituting
\eqref{eqinner} into \eqref{ind}, we obtain
$$
\mathbb{P}(\omega(\mathbf{X_{i}})\in\Omega)=c \sigma^{n-q-1}(\Omega).
$$
Therefore, we conclude that
$\omega(\mathbf{X_{i}})$ is uniformly distributed on $S^{n-1}\cap \mathcal{R}(R)^\perp$.\end{proof}
\smallskip

Next, we will show that $\mathbf{\tilde{X}}$, with high probability,
is an appropriate compressive sampling matrix, and thus we can use the
proposed clustering removal algorithm (CRA) to obtain an accurate
estimate of $\beta$ in \eqref{prob}. To that end, let \mbox{$U^T:\R^n\mapsto
\R^n$} be a
unitary transformation such that $U^T(\mathcal{R}(R)^\perp) =
\text{span}\{e_1,\dots,e_{n-q}\}$ where $e_i$ are the standard
basis vectors. Then we can write
$$U^T\widetilde{\mathbf{X}}=
\begin{bmatrix}
  \textbf{A}\\
  0_{q\times p}
\end{bmatrix}$$ where $\textbf{A}$ is a $(n-q)\times p$ random
matrix. Due to the rotation invariance of the Lebesgue measure, and since
$U$ is a deterministic matrix, the columns of $U^T\widetilde{\textbf{X}}$
are independently uniformly distributed over 
$$U^T (S^{n-1}\cap \mathcal{R}(R)^\perp)=\{x\in S^{n-1}: \ 
x_{n-q+1}=...=x_n=0\}.$$ This means that the columns of $\mathbf{A}$
are uniformly distributed over $S^{n-q-1}$. Theorem~\ref{th11} and
Theorem~\ref{th22} provide sufficient conditions for $\textbf{A}$ to
satisfy RIP with overwhelming probability. More precisely, the
following holds. 

\begin{cor}\label{th:RIP}
If $(n-q)\geq C\delta^{-2}k\log(\frac{ep}{k})$, then with probability
at least $1-2\exp(-c\delta^{2}(n-q))$, \textup{$\delta_{k}(\textbf{A})\leq\delta$\label{cor:obsnumber}.} 
\end{cor}

Next, noting that $\widetilde{\mathbf{X}} = U \begin{bmatrix}
  \textbf{A}\\
  0_{q\times p}
\end{bmatrix}$, we relate the RIP constants of
$\mathbf{\widetilde{X}}$ with the RIP constants of $\mathbf{A}$.


\begin{prop}\label{th:RIP-transition}
If $A$ satisfies the restricted isometry property with constant $\delta_{k}$, $U\begin{bmatrix}
A\\
0_{q\times p}
\end{bmatrix}$ will also satisfy RIP with the exact same constant.
\end{prop}
\begin{proof}
First, recall that a matrix $Z\in \R^{n\times p}$ satisfies RIP of order $k$ with
constant $\delta_k$ if and only if  for all $I \subset [p]$ with
$|I|=k$, the eigenvalues of $Z^T_IZ_I$ are all between $1-\delta_k$
and $1+\delta_k$. Here $Z_I$ denotes the submatrix of $Z$ consisting
of its columns indexed by $I$. 

Suppose that $A$ satisfies RIP of order k with constant
$\delta_k$. Set
$$\widetilde{X}=U \begin{bmatrix}
  A\\
  0_{q\times p}
\end{bmatrix}
$$
and let $I$ be any subset of $[p]$ such that $|I|=k$. Then, 
\begin{align}
\widetilde{X}^T_I \widetilde{X}_I&= \Big(U\begin{bmatrix}
A_I\\
0_{q\times k}
\end{bmatrix}\Big)^T\Big(U\begin{bmatrix}
A_I\\
0_{q\times k}
\end{bmatrix}\Big) \notag \\
&= {\begin{bmatrix}
A_I\\
0_{q\times k}
\end{bmatrix}}^TU^TU \begin{bmatrix}
A_I\\
0_{q\times k}
\end{bmatrix}
\notag \\
&=A_I^TA_I. \notag
\end{align}

Therefore, for all $I\subset [p]$, eigenvalues of $\widetilde{X}_I^T\widetilde{X}_I$ and $A_I^TA_I$ are equal. Hence, $\widetilde{X}$ satisfies RIP of order $k$ with constant $\delta_k$ if $A$ satisfies RIP of order $k$ with the same constant. 
\end{proof}

Proposition~\ref{th:RIP-transition} suggests that if the assumptions
of Corollary~\ref{th:RIP} are met, $\widetilde{\textbf{X}}$ satisfies
RIP with overwhelming probability. This, in turn, provides uniform
estimation guarantees for $\gamma$, as defined in
\eqref{secondary}. In particular, combining Corollary~\ref{th:RIP} and
Proposition~\ref{th:RIP-transition}, we arrive at the following
conclusion.
\begin{thm}\label{main2}
Consider Problem~\ref{prob} and suppose that the data matrix $X$
satisfies Assumptions~\ref{a0}-\ref{a2}. 
\begin{enumerate}[(i)]
\item If $(n-q)\ge C (t-1)s
\log(\frac{ep}{ts})$ for some $t\ge 4/3$, then
$\delta_{ts}(\widetilde{X}) \le \sqrt{\frac{t-1}{t}}$ with probability
greater than \mbox{$1-\exp\big(-c \frac{t-1}{t}(n-q)\big)$}. 
\item Consequently, for any $s$-sparse $\gamma\in\mathbb{R}^{p}$, 
$$
||\gamma_{BPDN(\eta)}^{\#}-\gamma||_{2}\leq D\eta+\frac{2C\sigma_{s}(\gamma)_{\ell_1}}{\sqrt{s}}
$$
where $D, C$>0 are constants depending only on $\delta_{ts}$. Here
$\gamma$ is as in \eqref{secondary}.
\item The estimate in (ii) implies the following bound for the estimate of $\beta$: 
$$||\hat{{\beta}}-\beta||_{2}\leq\max_{i}\dfrac{1}{||\Pi_{R^{\perp}}X_{i}||_{2}}\big(D\eta+\dfrac{2C\sigma_{1}(\gamma)}{\sqrt{s}}
\big). 
$$
\end{enumerate}
\end{thm}

\begin{remark}
 Note that part (i) of Theorem~\ref{main2} is a restatement of
  Corollary~\ref{th:RIP} with $\delta=\sqrt{\frac{t-1}{t}}$ and
  $k=ts$. In particular, we can set $t=2$, which yields the sufficient
  condition
$$(n-q) \ge C s\log\big(\frac{ep}{2s}\big).
$$ 
Part (ii) of Theorem~\ref{main2} follows from
Theorem~\ref{th:RIP-transition} together with part (i). Part (iii) can be justified using the definition of $\gamma$ as well as Step 4 of Algorithm~\ref{alg1}.
\end{remark}



\begin{remark}
Part (iii) of Theorem~\ref{main2} shows that in the estimation of
$\beta$ the error term is inflated and behaves like
$\max(1/||\Pi_{R^{\perp}}X_{i}||_{2})$. However, it should be noted
that no matter what method is used, sparse recovery in highly
correlated settings is inherently unstable when the noise level is high. For
example, let $X_{i}$ and $X_{j}$ be unit-norm and highly correlated,
and let $T_{\beta}$ be the indices of the true support of
$\beta$. Assume $i\in T_{\beta}$ but $j\not\in T_{\beta}$.  Then,
without any assumptions on the structure of the noise $u$, as soon as
$||u||_{2}$ exceeds $||X_{i}-X_{j}||_{2}|\beta_{i}|$ it becomes
impossible to distinguish whether $i\in T_{\beta}$ or $j\in T_{\beta}$
by just considering $y$. If $X_{i}$ and $X_{j}$ are highly correlated,
then $\left\Vert X_{i}-X_{j}\right\Vert _{2}$ tends to be small and
unless $\beta_{i}$ is extremely large,
$||X_{i}-X_{j}||_{2}|\beta_{i}|$ tends to be small too. Hence, nearly
accurate sparse recovery is possible only for modest amounts of
$||u||_{2}$. As we will see in section \ref{sec:Numerical-Results},
the numerical results corroborate this result.
\end{remark}

\begin{remark}
  Our analysis in this section has been based on the (unrealistic)
  assumption that $\mathcal{R}(R)^\perp$ can be perfectly
  estimated. Let $E$ be an imperfect estimate of
  $\mathcal{R}(R)^\perp$. By some algebra it can be shown that if $E$
  is reasonably accurate, one can bound the matrix norm of the
  difference of $\widetilde{X}$ and $\widehat{X}:=\Pi_EX
  N_{\Pi_EX}^{-1}$. From there on we refer the reader to the analysis
  of \cite{herman2010general} which suggests that $\widehat{X}$ will
  still satisfy the restricted isometry property but with different
  constants, which depend on the restricted isometry constants of
  $\widetilde{X}$ and the accuracy of $E$. We also note that the
  numerical experiments we present in the next section use design
  matrices for which clusters or $\mathcal{R}(R)^\perp$ are not known
  a priori and can only be imperfectly estimated. Yet, the empirical results
  illustrate that CRA is still effective in these examples.  
\end{remark}

\section{Empirical Performance\label{sec:Numerical-Results}}

For the purpose of testing the empirical performance of our algorithm,
we provide two sets of experiments. In the first one, we generate
$X$ synthetically from a factor model. In the second experiment, we use correlated stock price time-series as the columns of $X$. In
each experiment, we run multiple trials in which we randomly choose
a $20-$sparse $\beta$ such that every non-zero entry is uniformly
distributed in $[1,2]$. Using this chosen $\beta$, we generate a
$y$ vector and test how well our algorithms can estimate $\beta$
given $y$ and $X$. We report the average performance of each algorithm
across all trials. 

Following the suggestion of \cite{belloni2013}, beside CRA, basis
pursuit, and SWAP, we add two new estimators: CRA-OLS, and BPDN-OLS.
In CRA-OLS and BPDN-OLS, we use the supports of $\hat{\beta}_{CRA}^{\{20\}}$
and $\hat{\beta}_{BPDN}^{\{20\}}$ as the support estimate of $\beta$.
Let $\hat{{T}}_{\beta}$ be this estimated support. We estimate the
non-zero entries of $\beta$ by $X_{\hat{T}_{\beta}}^{\dagger}y$.

When generating $y$ from $\beta$, we choose the noise vector
\inputencoding{latin1}{$u$ according to $\ensuremath{\alpha
    N(0_{n\times1},I_{n\times n})}$}\inputencoding{latin9}
distribution, where $\alpha$ is chosen such that the SNR is equal to
$\sigma$ dB with $\sigma\in\{10,15,20,...,100\}$.  For the
initialization step of SWAP and the sparse recovery step of CRA, we
use BPDN, with $||u||_{2}$ as its parameter. The solver that we use
for basis pursuit is SPGL1 \cite{BergFriedlander:2008,spgl1:2007}. For
each noise level, we generate 30 realizations of $\beta$ and as
our first performance measure, we report the average over these trials of
$\dfrac{||\beta-\hat{\beta}^{\{20\}}||_{2}}{||\beta||_{2}}$ across
different noise levels, for all five algorithms.

To empirically test the support recovery performance of CRA, we next
define the true positive rate of $\hat{\beta}$ to be
$$\TPR(\hat{\beta}):=\dfrac{|\text{supp}(\beta)\cap
  \text{supp}(\hat{{\beta}|}}{|\text{supp}(\beta)|}.
$$
In support recovery literature, the goal is to have a sparse
$\hat{\beta}$ such that $\TPR(\hat{\beta})$ is as close to one as
possible. As our second measure of performance, we report the average
$\TPR$ over 30 trials (as described above) of CRA, BPDN, and SWAP
across different noise levels. In noisy settings, results of basis
pursuit and CRA have a large number of non-zero entries.  Therefore,
although $\hat{\beta}_{CRA}$ and $\hat{\beta}_{BPDN}$ contain a large
proportion of the true support, they are not informative for
identifying it. To have a reasonable measure of performance for CRA
and BPDN, we report $\TPR(\hat{\beta}_{CRA}^{\{20\}})$ and
$\TPR(\hat{\beta}_{BPDN}^{\{20\}})$.

\subsection{Synthetic Data}\label{sec_A}
In our first experiment, we aim to test the performance of our algorithm on a synthetically generated data. Inspired by factor models, which are widely used in econometrics
and finance (we refer to \cite{bai2008large} for a survey on factor models and their applications), we generate the data matrix, $X$, 
according to
\[
X=F\Lambda+Z
\]
where $F\in\mathbb{R}^{n\times q}$, $\Lambda\in\mathbb{R}^{q\times p}$,
$Z\in\mathbb{R}^{n\times p}$, and $q\ll n$. The columns of $F$ are
the underlying factors that drive the cross-correlation among the columns
of $X$. $\Lambda$ is the coefficient matrix which determines the
extent of the exposure of each column of $X$ to the $F_{i}$'s. $Z$
is the idiosyncratic variation in the data, which in this experiment
is generated according to an i.i.d Gaussian distribution. We generate the $F_{i}$
according to an $\mbox{ARMA(2,0)}$ model as follows: 
\[
F_{t,i}=0.5F_{t-1,i}+0.3F_{t-2,i}+v_{t},\; v_{t}\sim N(0,1)
\]
This formulation ensures that entries of $X$ are reasonably correlated
both across rows and columns and therefore, it gives the data a rich correlation structure. We choose $n=250$, $p=1000$, and $q=25$. 

Even though the data is not explicitly decomposed of clusters
distributed uniformly on spherical caps, we can still apply CRA to
this problem. In this case, the columns of $F$ will play the role of
cluster centers. Due to the factor structure of the data, and since we
only need to estimate the $\mathcal{R}(F)$ to construct
$\widetilde{X}$, we use the range of the most significant $q$
eigenvectors of $XX'$ for estimating $\mathcal{R}(F)$\footnote{We
  can also use k-means clustering for this purpose and the estimation
  results are very similar.}. Figure \ref{fig:svd-compare-synthetic}
demonstrates the singular values of $X$ on the left hand side, and
singular values of $\widetilde{X}$ on the right. For comparison
purposes, we have also plotted the singular values of a $250\times
1000$ matrix with columns uniformly distributed on $S^{224}$. As we
can see, the singular values of $\widetilde{X}$ closely resemble those
of the matrix with uniformly distributed columns.

\begin{figure}[t]
  \centering \includegraphics[height=6.5cm]{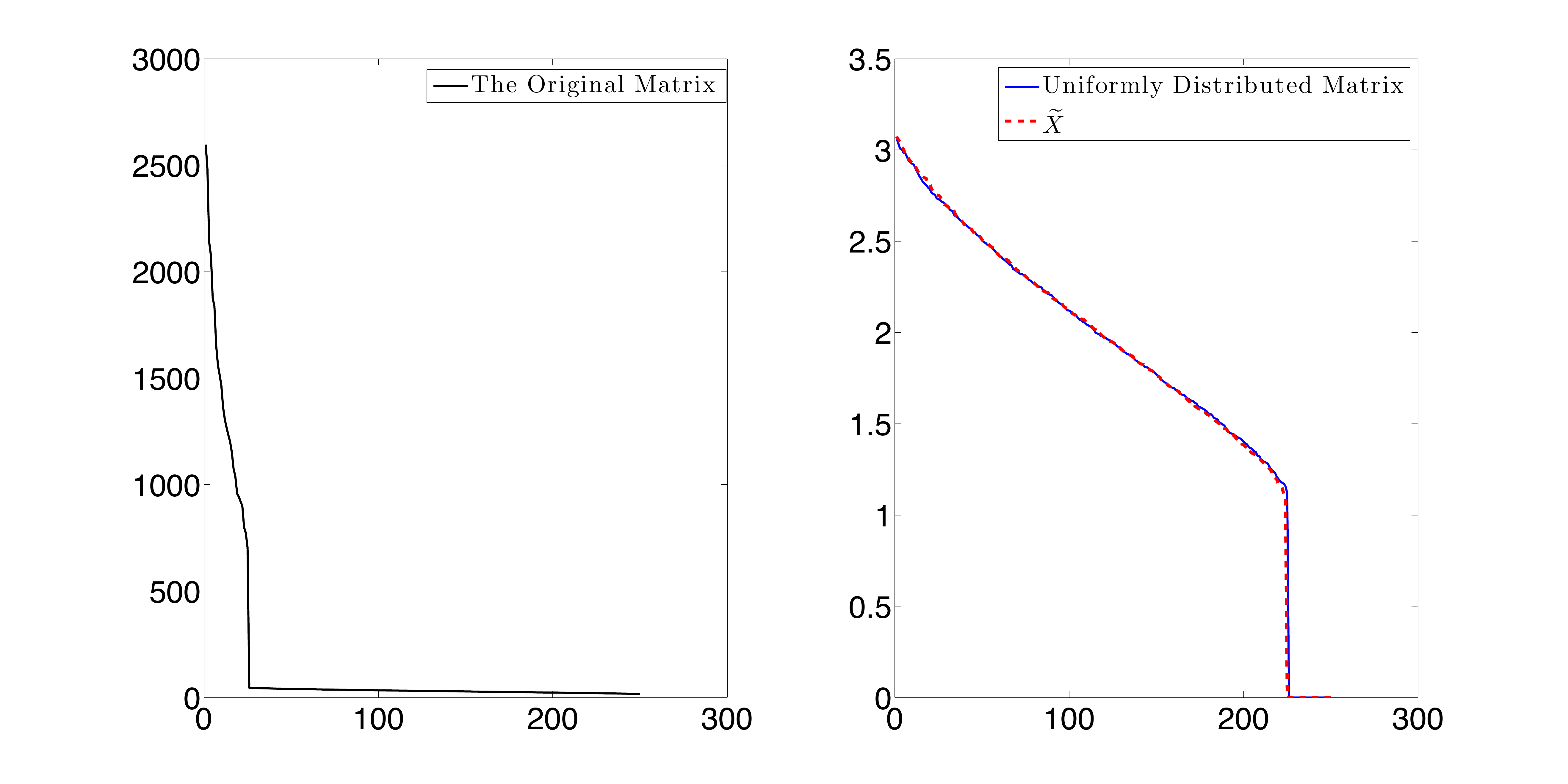}
  \protect\caption{The SVD comparison of $\widetilde{X}$, $X$, and a
    uniformly distributed matrix in the ``synthetic data'' experiment of
  Section~\ref{sec_A}. \label{fig:svd-compare-synthetic}}
\end{figure}
 
Figure \ref{fig:Estimation-err-synthetic} shows the average value of
$\dfrac{||\beta-\hat{\beta}^{\{20\}}||_{2}}{||\beta||_{2}}$ for all of
our algorithms, across various noise levels. As Figure
\ref{fig:Estimation-err-synthetic} demonstrates, for SNR levels more
than $30$ dB, CRA-OLS and CRA have the best performance. For SNR
levels less than $30$ dB, the estimation quality deteriorates
significantly to the extent that no algorithm provides a reasonable
estimate of $\beta$.

\begin{figure}[t]
\centering \includegraphics[height=6.5cm]{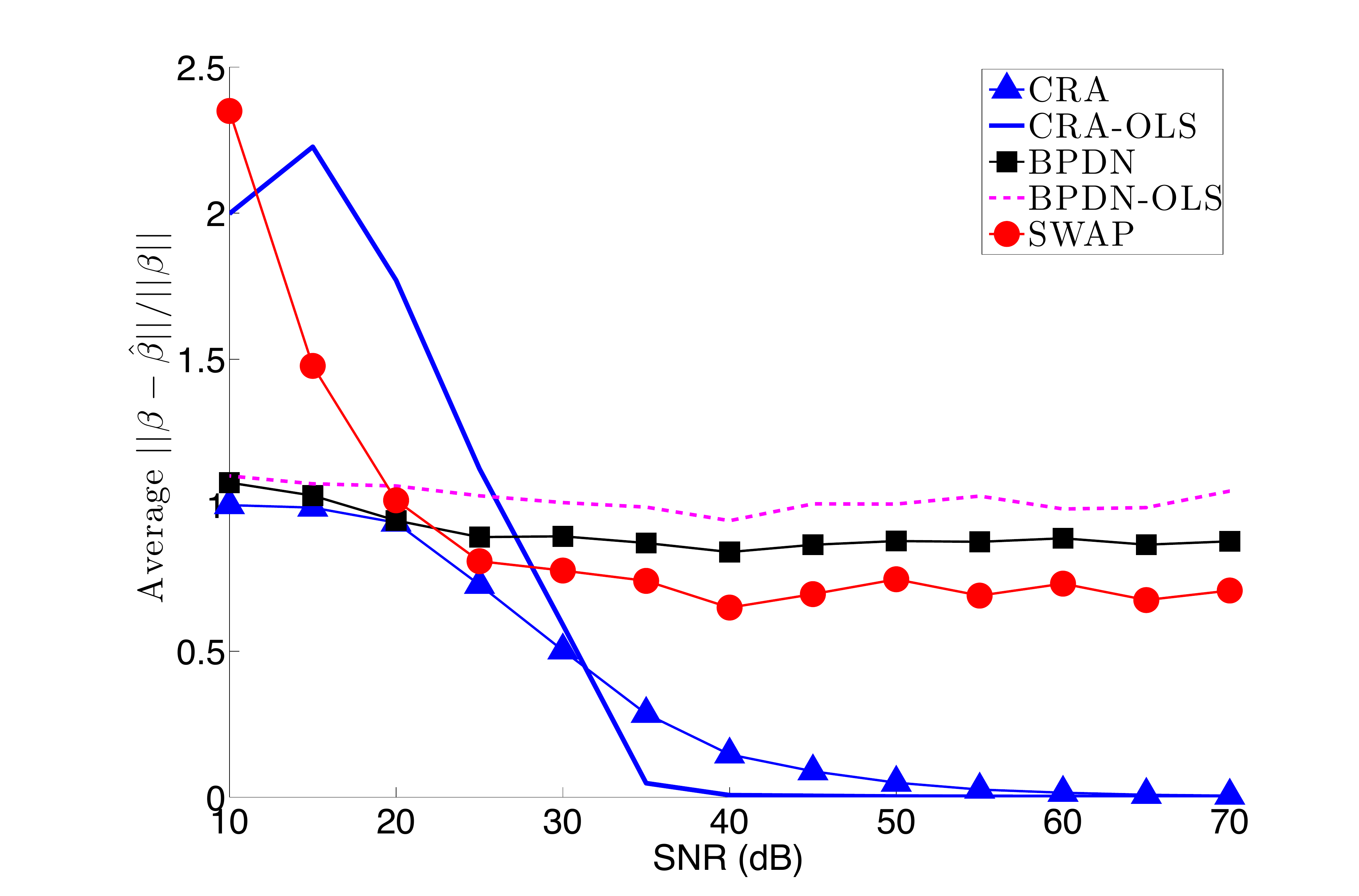}
\protect\caption{The average estimation error across different noise
  levels in the ``synthetic data'' experiment of
  Section~\ref{sec_A}. Each data point corresponds to the average
  estimation error over 30 realizations. \label{fig:Estimation-err-synthetic}}
\end{figure}

Figure \ref{fig:Estimation-tpr-synthetic} compares the true positive rate of the algorithms. As the figure suggests, CRA recovers most of the true support up until $30$ dB. Afterwards, the estimation quality falls significantly for all of the algorithms. 

\begin{figure}[t]
  \centering \includegraphics[height=6.5cm]{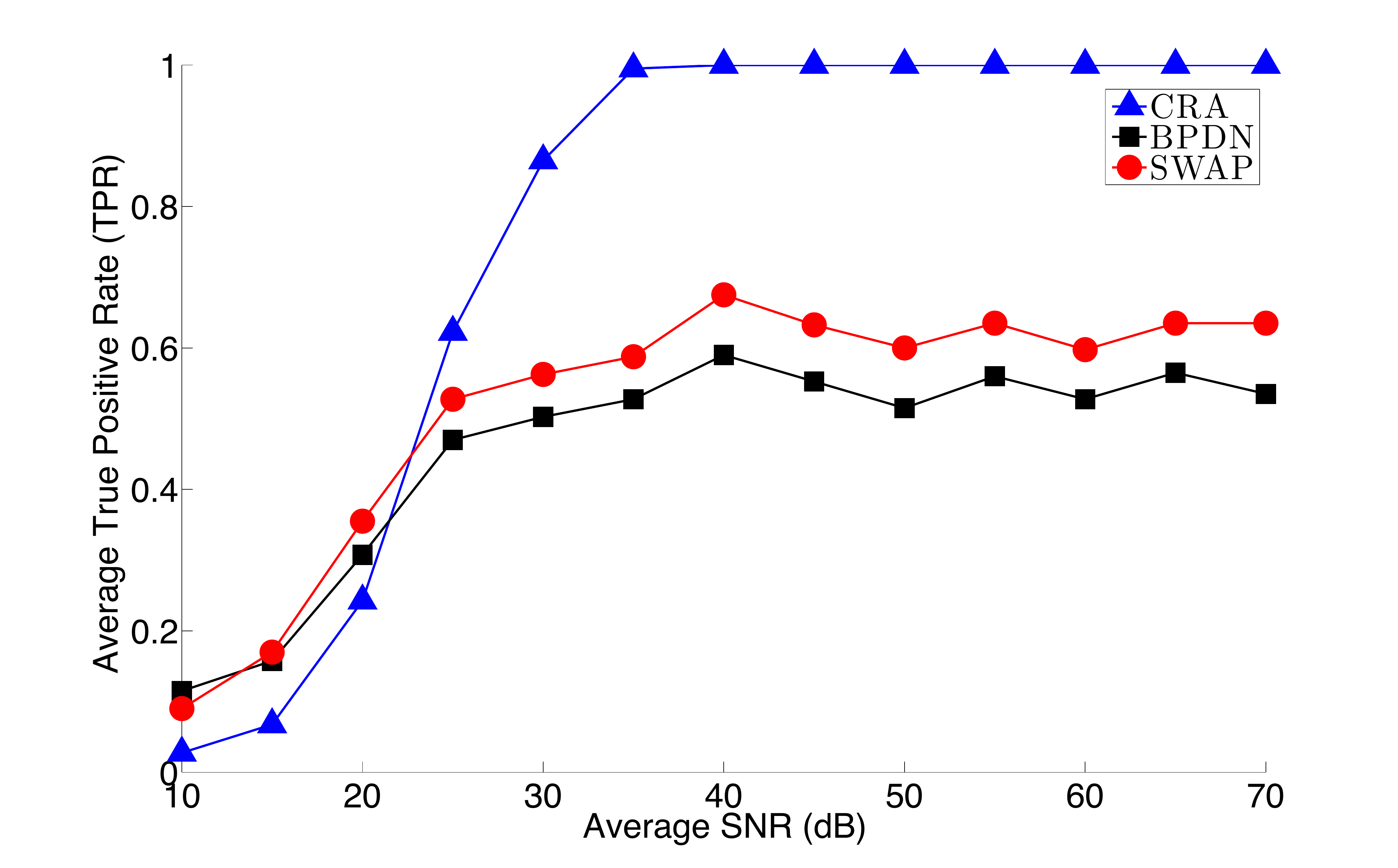}
  \protect\caption{The true positive rate across different noise
    levels in the ``synthetic data'' experiment of
    Section~\ref{sec_A}. Each data point corresponds to the average
    true positive rate over 30
    realizations. \label{fig:Estimation-tpr-synthetic}}
\end{figure}

The average running time of the algorithms, all on the same machine,
are presented in Table~\ref{t1}. Empirically, we observe that SPGL1 solves the CRA problem, which does not have high correlation in it, much faster. Moreover both CRA and BPDN are considerably faster than SWAP. 
\begin{table}[h]
\begin{center}
\begin{tabular}{|l|r|}
\hline 
Algorithm  & Running Times (s)\tabularnewline
\hline 
\hline 
CRA (without clustering)  & 0.044\tabularnewline
\hline 
clustering with eigenvectors  & 0.023 \tabularnewline
\hline 
BPDN  & 0.730\tabularnewline
\hline 
SWAP   & 4.100 \tabularnewline
\hline 
 
\end{tabular}
\par\end{center}
\caption{\label{t1} Average running times of various algorithms over
  570 trials (30 trials per noise level for 19 different values of $\sigma$)  for the example
  presented in Section~\ref{sec_A}.}
\end{table}

\subsection{Pseudo-Real Data} \label{sec_B}
\subsubsection{The Data Description}

For the purpose of the numerical simulations with real data, we consider
daily stock prices. Daily stock prices present us with a real life
data set with a complex high correlation environment. We use the prices
of 2000 NASDAQ stocks recorded from June 15 2012 to June 15 2014%
\footnote{The data is downloaded from Yahoo! Finance%
}. After removing the holidays, non-trading days, and special trading
session, each stock has 500 observations. To give some structure to
the data, we remove the trends from each time series, and also rescale
the $l_{2}$ norm of each series to $\sqrt{n}$. After this stage,
the data is used to populate the matrix $X\in\mathbb{R}^{500\times2000}$.
Figure \ref{fig:The-Graphical-Representation} is the graphical representation
of $\dfrac{1}{n}X^{*}X$. It is evident that, even after all modification
on the data, $\dfrac{1}{n}X^{*}X$ is very different from identity,
and hence $X$ does not fit into the category of suitable matrices
for sparse recovery.

\begin{figure}[t]
\centering

\includegraphics[width=9cm]{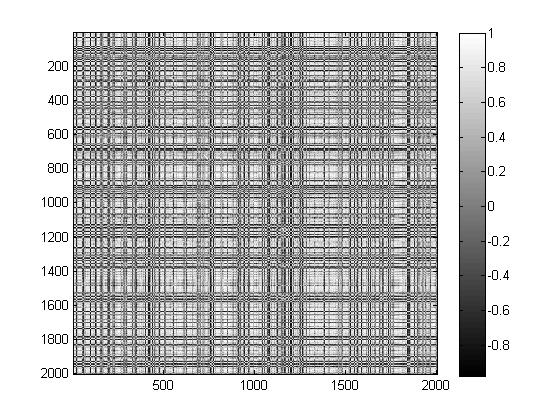}

\protect\protect\caption{The matrix $\frac{1}{n}X^{*}X$. \label{fig:The-Graphical-Representation}}
\end{figure}

\subsubsection{Results}

By using k-means clustering, we identify 14 clusters in the data, hence $R\in\mathbb{R}^{500\times14}$.
Figure \ref{fig:The-Covariance-Matrix} is the graphical representation
of $\frac{1}{n}\left(\Pi_{R^{\perp}}XN_{\Pi_{R^{\perp}}X}^{-1}\right)^{*}\left(\Pi_{R^{\perp}}XN_{\Pi_{R^{\perp}}X}^{-1}\right)$.
As we can see, we have a tremendous improvement in the covariance
structure of the data matrix.

\begin{figure}[t]
\centering \includegraphics[width=9cm]{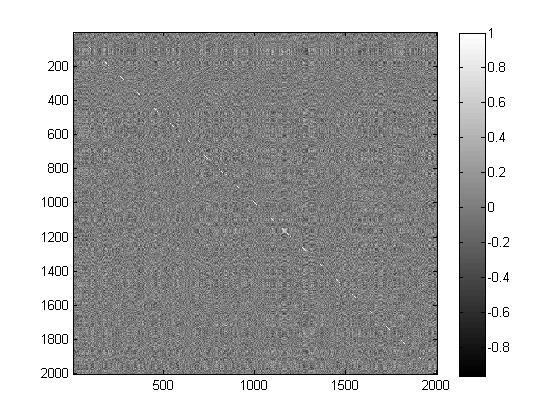}

\protect\protect\caption{The covariance matrix of $\Pi_{R^{\perp}}XN_{\Pi_{R^{\perp}}X}^{-1}$
\label{fig:The-Covariance-Matrix}}
\end{figure}

Figure \ref{fig:svd-real-comparison} shows the singular values of $X$
on the left hand side, and singular values of $\widetilde{X}$ on the
right. For comparison we include in the left figure the plot of the
singular values of a $500\times 2000$ matrix with columns uniformly
distributed on $S^{484}$. Although there is considerable difference
between the singular values of $\widetilde{X}$ and those of the matrix
with uniformly distributed columns, we observe a vast improvement
compared to the original case. This improvement is in fact sufficient
for us to successfully perform sparse estimation. One can choose more
than 14 clusters or use a more advanced clustering method to make the
singular values of $\widetilde{X}$ closer to the uniform
case. However, we empirically observed that this does not improve the
estimation results significantly.

\begin{figure}[t]
\centering \includegraphics[height=6.5cm]{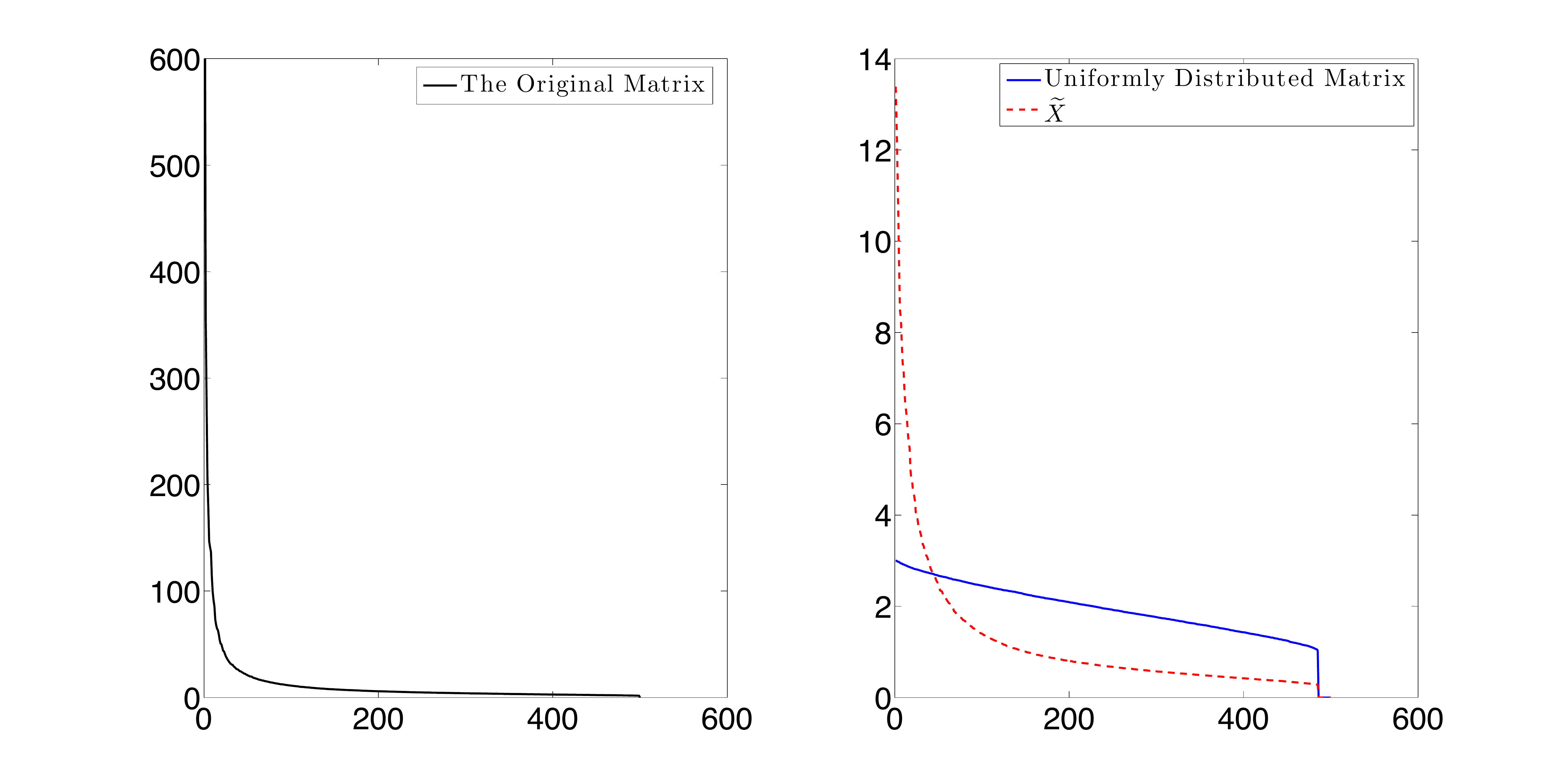}

\protect\protect\caption{The SVD comparison of $\widetilde{X}$, $X$, and a
    uniformly distributed matrix in the ``pseudo-real data'' experiment of
  Section~\ref{sec_B}.
\label{fig:svd-real-comparison}}
\end{figure}

Figure \ref{fig:The-Average-Estimation-err} demonstrates the average
value of $\dfrac{||\beta-\hat{\beta}^{\{20\}}||_{2}}{||\beta||_{2}}$
for all five algorithms. As Figure
\ref{fig:The-Average-Estimation-err} shows, in low noise settings,
CRA-OLS and CRA have the best performance among the algorithms. As the
noise level increases, distinguishing between the columns of $X$
becomes almost impossible. For noise levels more than $20$ dB, none of
the algorithms provide a reasonable estimate of $\beta$.

\begin{figure}[h]
\centering \includegraphics[height=6.5cm]{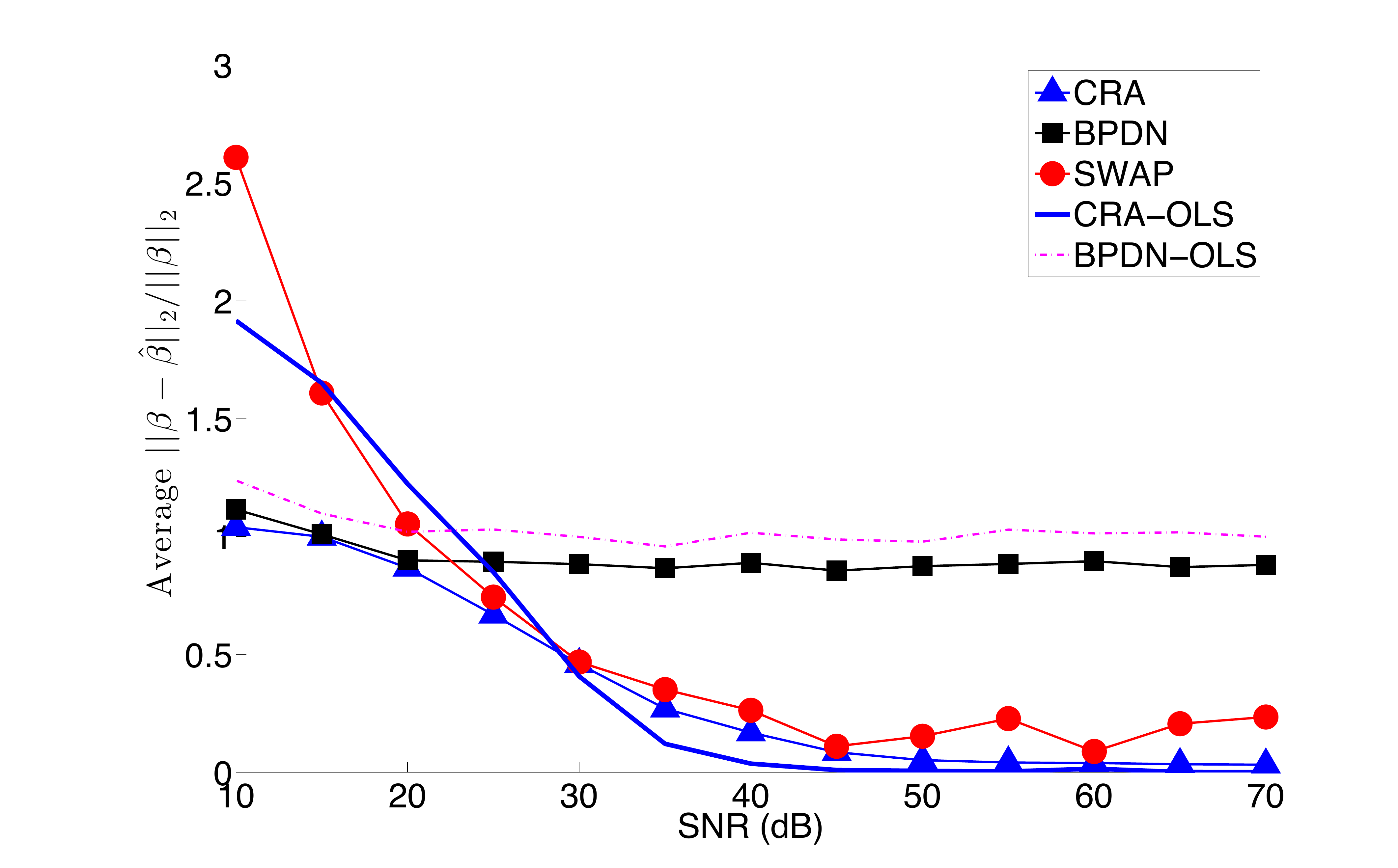}

\protect\caption{The average estimation error across different noise
  levels in the ``pseudo-real data'' experiment of
  Section~\ref{sec_B}. Each data point corresponds to the average
  estimation error over 30 realizations. \label{fig:The-Average-Estimation-err}}
\end{figure}

Figure \ref{fig:The-Average-True-Positive Rate} shows the average
value of $\TPR(\hat{\beta}^{\{20\}})$ across different noise levels,
for each algorithm. It can be observed that again CRA outperforms
both SWAP and BPDN.

\begin{figure}[h]
\centering \includegraphics[height=6.5cm]{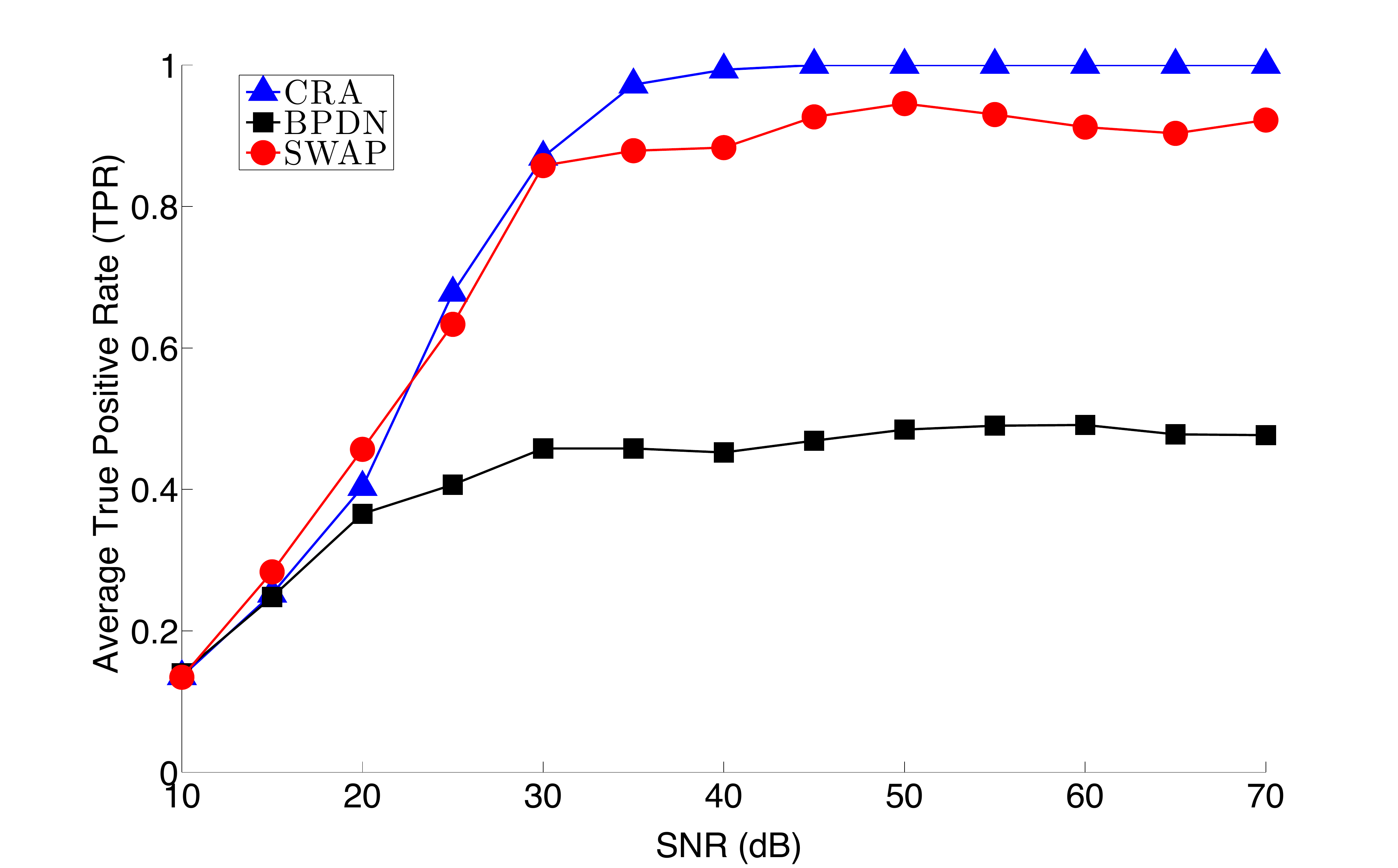}
\protect\caption{The true positive rate across different noise
    levels in the ``pseudo-real data'' experiment of
    Section~\ref{sec_B}. Each data point corresponds to the average
    true positive rate over 30
    realizations.\label{fig:The-Average-True-Positive Rate}}
\end{figure}

In SWAP, the sparsity rate of the final vector is always equal to
the size of the initial support provided, $|T_{0}|$. To make sure
that SWAP's estimate has a high TPR ( although in the expense of adding
false positive indices), \cite{vats2013swapping} suggest providing
a larger $T_{0}$. Indeed, our results suggest that if we allow SWAP
to search for a 30-sparse vector in estimating a 20-sparse $\beta$,
$\TPR(\hat{\beta}_{SWAP})$ increases. However, as the size of the
support that SWAP is working with increases, the running time of SWAP
increases significantly. The following table represent the average
running time of the algorithms on the same machine\footnote{The
  experiments presented in Sections~\ref{sec_A} and \ref{sec_B} are
  conducted on two different machines respectively, so the running
  times presented in Table~\ref{t1} and Table~\ref{t2} cannot be compared directly.}

\begin{table}[h]
\begin{center}
\begin{tabular}{|l|r|}
\hline 
Algorithm  & Running Times (s)\tabularnewline
\hline 
\hline 
BPDN  & 0.90 \hskip 5cm\tabularnewline
\hline 
K-means Clustering  & 2.48 \hskip 5cm\tabularnewline
\hline 
CRA (with Clustering)  & 3.70 \hskip 5cm\tabularnewline
\hline 
SWAP (20-sparse)  & 18.57 \hskip 5cm\tabularnewline
\hline 
SWAP (30-sparse)  & 146.27 \hskip 5cm\tabularnewline
\hline 
\end{tabular}
\par\end{center}
\caption{\label{t2} Average running times of various algorithms over
  570 trials (30 trials per noise level for 19 different values of $\sigma$)  for the example
  presented in Section~\ref{sec_B}.}
\end{table}

A second point to consider is the sensitivity of the algorithm's performance
to the number of observations. We run the same experiment as above
but we discard the first 250 observations so $X\in\mathbb{R}^{250\times2000}$.
Figure \ref{fig:The-Difference} shows the difference between the
average TPR with 500 observations and the average TPR with 250 observations
for each algorithm. In comparison to other algorithms, CRA's performance
falls much less when the number of observations decreases while SWAP's
performance degrades significantly.

\begin{figure}[h]
\centering \includegraphics[height=6.5cm]{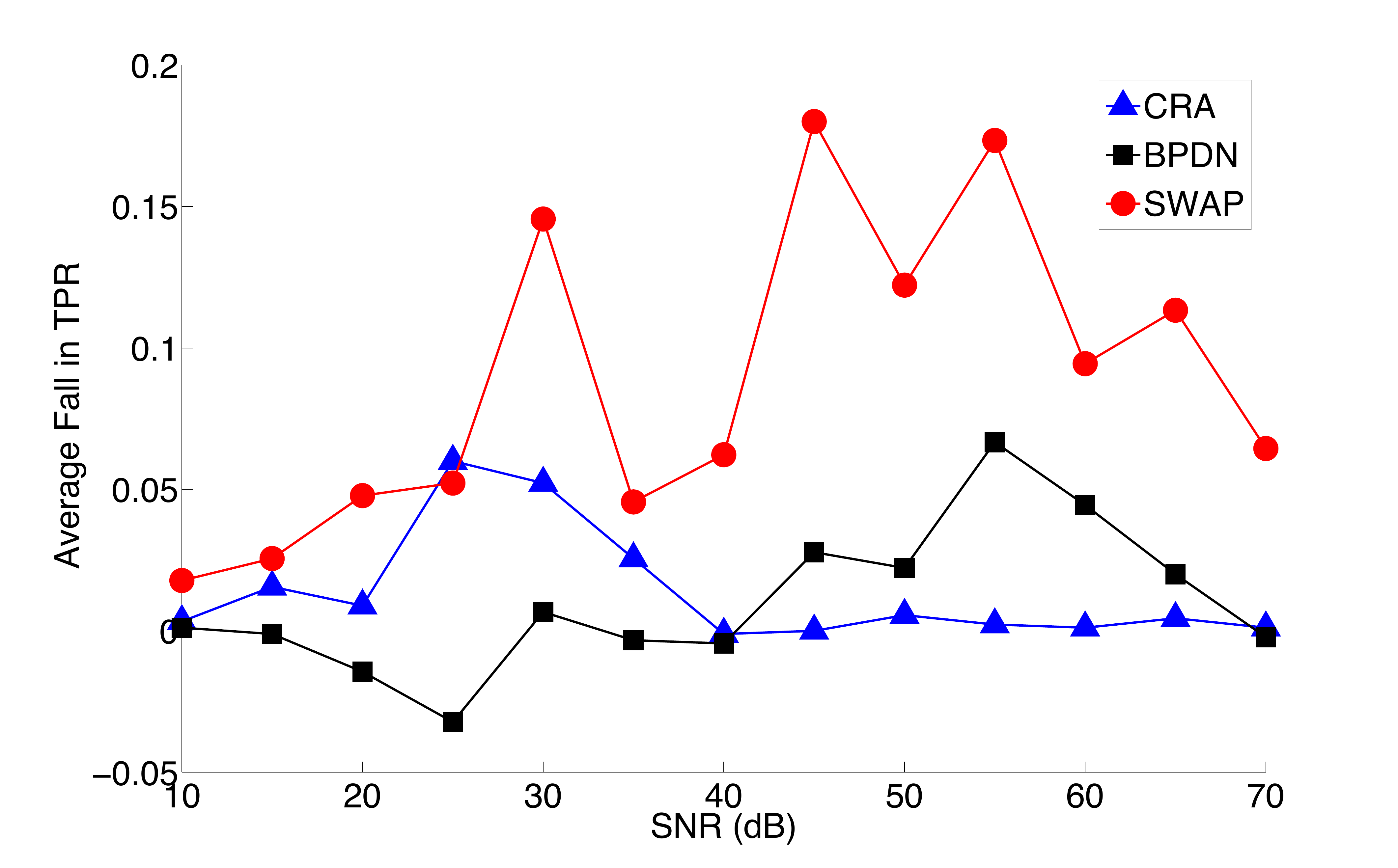}

\protect\caption{The difference between the performances of each
  algorithm when the number of observations is reduced from 500 to
  250.  \label{fig:The-Difference}}
\end{figure}

\section{Conclusions and Future Work}

CRA is a computationally efficient estimation method that allows for
sparse recovery in highly correlated environments. As we showed in
Section~\ref{sec:Numerical-Results}, even with extremely simple
clustering algorithms, CRA's empirical performance is superior to
other state-of-the-art sparse recovery algorithms. Moreover, CRA is
computationally efficient (as reflected in the run times given in
Tables~\ref{t1} and \ref{t2}) which makes it suitable for working with
large data sets. In addition, the algorithm provides a large degree of
flexibility in the choices of the sparse recovery technique and
clustering method. Accordingly, CRA is a versatile sparse estimation
method for cases where the design matrix (or equivalently, the
compressive sensing measurement matrix) has correlated columns that
can be grouped into a small number of clusters as we described
before. 

Despite all these strengths, there are certain setting where one may
need to modify CRA. First, as we noted earlier, in noisy environments,
as the data becomes more and more highly correlated, distinguishing
the columns from each other becomes impossible. In this case, to
decrease the correlation in the matrix, one might want to group the
columns that are extremely correlated with each other and then run CRA
on the representative vectors of the grouped variables. There are many
approaches for variable grouping, e.g., \cite{2012arXiv1209.5908B} and
\cite{OWL} introduce two new automated grouping algorithms and provide
a short survey of the field. Combining CRA with these variable
grouping algorithms is an interesting subject for future research.

Another interesting problem is whether one can relax the rigid cluster
structure imposed by our assumptions. This would mean that we adapt
CRA for sparse recovery in the case of design matrices with correlated
columns that do not satisfy our Assumptions~\ref{a0}-\ref{a2}. 

\section{Acknowledgements} B. Ghorbani was funded by a UBC Arts Summer
Research Award and a UBC Arts Undergraduate Research
Award. {\"O}.~Y{\i}lmaz was funded in part by a Natural Sciences and
Engineering Research Council of Canada (NSERC) Discovery Grant
(22R82411), an NSERC Accelerator Award (22R68054) and an NSERC
Collaborative Research and Development Grant DNOISE II (22R07504).

 \bibliographystyle{plain}
\bibliography{References}

\end{document}